\documentclass[10pt]{amsart}

\usepackage{amsmath,amssymb}
\usepackage{epsfig,url}
\usepackage{mathdots}
\usepackage{mathrsfs}
\usepackage{cite}
\usepackage{color}

\usepackage{blindtext}
\usepackage{bm}
\usepackage{booktabs}
\usepackage{multirow}

\usepackage[colorlinks=true,urlcolor=red,citecolor=red]{hyperref}
\usepackage[left=3cm,right=3cm,top=3cm,bottom=3cm]{geometry}

\usepackage{verbatim}
\newtheorem{theorem}{Theorem}[section]

\theoremstyle{definition}

\newtheorem{note}[theorem]{Note}

\theoremstyle{remark}

\newtheorem{Definition}{\bf Definition}[section]

\newtheorem{Example}[Definition]{\bf Example}

\numberwithin{equation}{section}

\newcommand{\paren}[1]{\left(#1\right)}

\newcommand{\ba}{\begin{eqnarray}}
\newcommand{\ea}{\end{eqnarray}}
\newcommand{\ift}{\int_{0}^{\infty}}

\newcommand{\pFq}[5]{\ensuremath{{}_{#1}F_{#2} \left( \genfrac{}{}{0pt}{}{#3}
{#4} \bigg| {#5} \right)}}

\newmuskip\pFqmuskip

\newbox{\ORCIDicon}
\sbox{\ORCIDicon}{\large\includegraphics[width=0.8em]{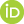}}

\newcommand{\orcid}[1]{\,\href{#1}{\usebox{\ORCIDicon}}\,}

\begin{document}

\title{
New formula for Asymptotic behavior of the Synchrotron function
}

\author[Iv\'{a}n Gonz\'{a}lez]{Ivan Gonzalez$^1$\orcid{https://orcid.org/0009-0002-0731-9519}}
\email{ivan.gonzalez@uv.cl}

\author[Daniel Salinas-Arizmendi]{Daniel Salinas-Arizmendi$^2$\orcid{https://orcid.org/0000-0002-0577-2005}}
\email{daniel.salinas@usm.cl}

\dedicatory{${}^1$Instituto de {F}\'{i}sica y {A}stronom\'{i}a,~{U}niversidad de {V}alpara\'{i}so,~{A}venida~{G}ran~{B}reta\~{n}{a} $1111$,~Valpara\'{i}so,~{C}hile\\
${}^2$Departamento de Física, Universidad Técnica Federico Santa María y Centro Científico Tecnológico de Valparaíso, Casilla 110-V, Valparaíso, Chile}

\date{\today}

\begin{abstract}
Synchrotron radiation plays a central role in astrophysical and high-energy processes. Its spectral description involves the synchrotron function, defined by a non-trivial integral of modified Bessel functions and commonly evaluated through numerical methods or dedicated approximations. In this work, we obtain a compact analytical representation of the synchrotron function using the \textit{Method of Brackets}, which yields systematically controllable asymptotic expansions in both the small- and large-argument regimes. The resulting expressions accurately reproduce numerical integration and make the analytic structure of the function explicit. Our results provide an efficient alternative to repeated numerical evaluations and facilitate applications requiring fast and controlled approximations.
\end{abstract}

\maketitle

\section{Introduction}\label{sec:introduction}

Synchrotron radiation is a fundamental mechanism in astrophysical and high-energy environments, generated by charged particles accelerated to relativistic velocities in the presence of magnetic fields. In standard formulations, several quantities associated with synchrotron emission are expressed in terms of special functions, in particular the modified Bessel functions of the second kind $K_\nu$, which arise naturally within the classical theory of synchrotron radiation \cite{jackson1962,rybicki1979,schwinger1949,westfold1959}.  
Specifically, the synchrotron function can be written as
\begin{equation}\label{eq:integral}
F(x)= x \int_{x}^{\infty} K_{5/3}(z)\ dz,
\end{equation}
which appears in the modeling of synchrotron radiation spectra from a wide variety of sources, including pulsar wind nebulae, gamma-ray bursts, and active galactic nuclei, as well as in studies of relativistic electrons in strong magnetic fields near black holes and in turbulent magnetic field environments \cite{Yan:2023uxd,Chael:2017ahn,Derishev:2019nvj,Fang:2014joa}.

Despite its ubiquity, the integral definition \eqref{eq:integral} is not always the most convenient form for repeated evaluations, as it involves an integrand given by a special function and an improper upper limit. In practical applications, it is common to rely on numerical quadratures, dedicated approximations, or fitted formulas (\emph{fits}) tailored to specific parametrizations \cite{fouka2009,fouka2014analytical,Fouka:2013yi}. These approaches have proven effective in many contexts; however, in large-scale computations or asymptotic studies, it can be advantageous to have analytical representations that make explicit the structure of the $x\ll 1$ and $x\gg 1$ regimes, which are frequently relevant in physical applications.

In this work, we obtain an analytical representation of the synchrotron emission integral \eqref{eq:integral} by applying the \texttt{Method of Brackets} (MoB) \cite{GONZALEZ201050,brackets2-2010}. The MoB is a heuristic yet effective framework for reducing definite integrals to algebraic constraints on summation indices, and it has been successfully applied to a variety of problems involving special functions and definite integrals \cite{Gonzalez:2015msa,Gonzalez:2023jig,Gonzalez:2025elh}. By constructing a bracket series representation for $K_{5/3}$ and performing term-by-term integration within the MoB rules, we derive a compact analytical expression for $F(x)$ and extract systematically improvable asymptotic expansions for both $x\ll 1$ and $x\gg 1$, which are validated against direct numerical integration.

The remainder of the paper is organized as follows. In Section~\ref{sec:formalism}, we summarize the rules of the MoB employed throughout the manuscript. Section~\ref{sec:Bessel} develops the bracket series representations of $K_\nu$ and the divergent large-$x$ form relevant for the asymptotic analysis. In Section~\ref{sec:solution}, we derive an analytical expression for $F(x)$ and discuss its properties. Sections~\ref{sec:small} and~\ref{sec:large} present the small-$x$ and large-$x$ asymptotic regimes, respectively. Finally, Section~\ref{sec:conclusions} contains the conclusions.










\section{The Method of Brackets
}
\label{sec:formalism}

The
MoB
is a heuristic technique, 
powerful due to its simplicity and effectiveness, providing a straightforward approach for integrating terms involving special functions such as modified Bessel functions,
hypergeometric functions, etc. The method is based on transforming the multivariable integration problem into the solution of a linear system of equations, which can be easily solved.

The MoB evaluates multivariable definite integrals of the form:
\begin{equation}
J=\int_{0}^{\infty} \ldots \int_{0}^{\infty} f\left(x_{1}, \ldots, x_{n}\right) d x_{1} \ldots d x_{n},
\end{equation}
and originates in quantum field theory, specifically in the evaluation of loop integrals arising from Feynman diagrams, as an optimization and generalization of the Negative Dimensional Integration Method
\cite{Halliday:1987an,Anastasiou:1999ui,Suzuki:2002ak}.
The introduction of MoB as an integration technique enabled the resolution of more complex Feynman diagrams \cite{Gonzalez:2007ry,Gonzalez:2008uw,Gonzalez:2008uf,Gonzalez:2009nt} and its extension to integrals beyond those associated with Feynman diagrams \cite{gonzalez2014evaluation,GONZALEZ201050,brackets2-2010,Gonzalez:2011nq}. Finally, it is important to note that MoB serves as a multivariable generalization of Ramanujan’s Master Theorem \cite{cite-key}.

\begin{Definition}
The \texttt{Bracket} is the fundamental structure defined as
\begin{equation}
  \ift  x^{\alpha-1} d x \equiv \langle\alpha\rangle.
\end{equation}
\end{Definition}

The assignment to the divergent integral is valid for $\alpha \in \mathbb{R}$. The brackets resummation corresponds to the \texttt{bracket series}, $S(\alpha, \beta)$, which can be expressed as:
\begin{equation}\label{eq:seriebrackets2}
S(\alpha,\beta)=\sum_{n \geq 0} \phi_{n} \mathcal{F}(n)\langle\alpha n+\beta\rangle,
\end{equation}
where $\mathcal{F}(n)$ is a coefficient of the series, $\alpha,\ \beta$ are real parameters and the symbol  
$\phi_{n}$ is defined by
\begin{equation}
 \phi_{n} \equiv \frac{(-1)^{n}}{\Gamma(n+1)},
\end{equation}The following describes the formal rules for operating with these brackets.

\noindent \textbf{R1. Summation rule}. 
The bracket series in Eq.~\eqref{eq:seriebrackets2} is assigned the following value:
\begin{equation}\label{eq:bracketserie}
\sum_{n \geq 0} \phi_{n} \mathcal{F}(n)\langle\alpha n+\beta\rangle=\frac{1}{|\alpha|} \Gamma\left(\frac{\beta}{\alpha}\right) \mathcal{F}\left(-\frac{\beta}{\alpha}\right),
\end{equation}

The proof of the summation rule can be found in Refs. \cite{GONZALEZ201050,brackets2-2010,Gonzalez:2007ry}.

In the general case, the \texttt{multi-dimensional bracket series} version involves $r$ summations and $r$ brackets, and is given by:
\begin{equation}
\begin{aligned}
J = & \sum_{n_1 \geq 0} \cdots \sum_{n_r \geq 0} \phi_{n_1 \ldots n_r} \mathcal{F}\left(n_1, \ldots, n_r\right) \\
& \times \left\langle M_{11} n_1+\cdots+M_{1 r} n_r+c_1\right\rangle \ldots\left\langle M_{r 1} n_1+\cdots+M_{r r} n_r+c_r\right\rangle ,
\end{aligned}
\end{equation}
where $\phi_{n_1 \ldots n_r} = \phi_{n_1} \ldots \phi_{n_r}$ has been defined to simplify the notation.

By applying the summation rule iteratively, the general expression for $J$ is obtained
\begin{equation}\label{eq:sumrule2}
 J=\mathcal{F}\left(n_1^*, \ldots, n_r^*\right) \frac{\Gamma\left(-n_1^*\right) \ldots \Gamma\left(-n_r^*\right)}{|\operatorname{det}(\mathbf{M})|} ,  
\end{equation}
where the set of values $\left\{n_{i}^{*}\right\}(i=1, \ldots, r)$ is the solution to the system generated by setting the bracket arguments to zero, that is
\begin{equation}
\left(\begin{array}{c}
n_{1}^{*} \\
\vdots \\
n_{r}^{*}
\end{array}\right)=\left(\begin{array}{ccc}
M_{11} & \cdots & M_{1 r} \\
\vdots & \ddots & \vdots \\
M_{r 1} & \cdots & M_{r r}
\end{array}\right)^{-1}\left(\begin{array}{c}
-c_{1} \\
\vdots \\
-c_{r}
\end{array}\right)  ,  
\end{equation}
such that $\mathbf{M}$ is the matrix formed by the coefficients $\left\{M_{i j}\right\}(i, j=1, \ldots, r)$
\begin{equation}
 \mathbf{M}=\left(\begin{array}{ccc}
M_{11} & \cdots & M_{1 r} \\
\vdots & \ddots & \vdots \\
M_{r 1} & \cdots & M_{r r}
\end{array}\right)   .
\end{equation}

The value of Eq.~\eqref{eq:sumrule2} is undefined if the matrix $\mathbf{M}$ is not invertible. In cases where the number of summations exceeds the number of brackets, the procedure is detailed in \cite{GONZALEZ201050,brackets2-2010}.

\noindent \textbf{R2. Integration rule.} For the integral
\begin{equation}
    J=\int_{0}^{\infty} f(x) d x,
\end{equation}
the arbitrary function \( f(x) \) admits the following power series expansion
\begin{equation}
    f(x)=\sum_{n \geq 0} \phi_{n} \mathcal{F}(n) x^{\alpha n+\beta-1}.
\end{equation}

To apply the MoB, we will use the series representation of the function $f(x)$ in the integral, taking into account the definition of the brackets and the first rule, yielding
\begin{equation}
\int_{0}^{\infty} f(x) d x= \frac{1}{|\alpha|} \Gamma\left(\frac{\beta}{\alpha}\right) \mathcal{F}\left(-\frac{\beta}{\alpha}\right).
\end{equation}
This result corresponds to the integration rule of the MoB, applicable to a broad class of functions with bracket series representations~\cite{Gonzalez:2025elh}.

\noindent \textbf{R3. The multinomial expansion}. For multinomials, the expansion rule is given by the following expression:
\begin{equation}
\left(A_{1}+\ldots+A_{r}\right)^{ \pm \mu}=\sum_{n_{1}} \ldots \sum_{n_{r}} \phi_{n_{1} \ldots n_{r}}\left(A_{1}\right)^{n_{1}} \ldots\left(A_{r}\right)^{n_{r}} \frac{\left\langle\mp \mu+n_{1}+\ldots+n_{r}\right\rangle}{\Gamma(\mp \mu)}   
\end{equation}

\begin{proof}
To see how this is derived, see \cite{GONZALEZ201050}.
\end{proof}

\textbf{An extension of MoB.} By reducing an integral to its equivalent bracket series, the number of brackets obtained is equal to or smaller than the number of summations. To each bracket series representation, we associate a \texttt{complexity index}, defined as:
\begin{equation}
I_c = \sigma - \delta,
\end{equation}
where $ \sigma $ is the number of summations and $ \delta $ is the number of brackets. For $ I_c > 0 $, there are \hbox{$C_\delta^\sigma=\sigma!/\left[ \delta!\paren{\sigma-\delta}!\right]$} possible combinations to apply the summation rule \textbf{R1}. Each combination generates a term, corresponding to a power series with multiplicity $ I_c $. The solution to the integral is the sum of all the power series that share the same region of convergence.

\begin{Example} A generic integral of the form
\begin{equation}\label{eq:example}
J(A,B) = \ift f\paren{Ax}h\paren{Bx} dx,
\end{equation}
where the integrand can be separated into the following functions
\begin{eqnarray}
 f\paren{Ax} & = & \sum_n \phi_n \mathcal{F}(n) A^n x^n, \\
 h\paren{Bx} & = & \sum_\ell \phi_\ell \mathcal{H}(\ell) B^\ell x^\ell.
\end{eqnarray}
with
$A,\ B$ are parameters and $\mathcal{F}\left(n\right)$, $\mathcal{H}(\ell)$ are a functions of the indices.

Using the MoB rules, the integral in Eq.~\eqref{eq:example} has a bracket series of the form
\begin{equation}
J(A,B) = \sum_{n \geq 0}\sum_{\ell \geq 0} \phi_{n, \ell} \mathcal{F}(n) \mathcal{H}(\ell) A^{n} B^{\ell} \left\langle  n+\ell+1\right\rangle ,
\end{equation}

Applying the summation rule to eliminate the index $n$, we have
\begin{equation}
J_1(A,B)= \frac{1}{A} \sum_{\ell \geq 0} \paren{-1}^{\ell} \mathcal{F}(-\ell-1) 
\mathcal{H}(\ell) \paren{\frac{B}{A}}^{\ell},
\end{equation}
similarly, applying the summation rule to eliminate the index $\ell$, we obtaind a seconds term de la forma
\begin{equation}
J_2(A,B)= \frac{1}{B} \sum_{n \geq 0} (-1)^{n}
\mathcal{F}(n)\mathcal{H}(-n-1) \paren{\frac{A}{B}}^n.
\end{equation}

Both solutions have different regions of convergence, therefore the solution to the integral is:
\begin{equation}
J(A, B)=\left\{\begin{array}{l}
J_1(A, B),\\
J_2(A, B).
\end{array}\right.
\end{equation}

If the resulting series have a finite and nonzero radius of convergence $R$,
then, if $J_{1}(A,B)$ is valid for $\left\vert \frac{B}{A}\right\vert <R$,
it follows that $J_{2}(A,B)$ is valid for $\left\vert \frac{B}{A}%
\right\vert >R$. In this case, $J_{2}(A,B)$ corresponds to the analytic
continuation of $J_{1}(A,B)$.

On the other hand, if $J_{1}(A,B)$ has a radius of convergence $R\rightarrow
\infty $, then $J_{2}(A,B)$ converges only if $A=0$. In this case, the
series $J_{2}(A,B)$ is useful for obtaining asymptotic approximations when $%
\frac{A}{B}\rightarrow 0$.

\end{Example}

\section{\label{sec:Bessel}Brackets series representations of Bessel function}
In this section, we will demonstrate that the MoB method obtains asymptotic expansion for 
\(K_{\nu}(z)\),
yielding compact and analytical representation. These results are validated against well-known numerical methods.

The Bessel differential equation given by
\begin{equation} \label{eq:edo}
 x^{2} \frac{d^{2} y(x)}{d x^{2}}+x \frac{d y(x)}{d x}-\left(x^{2}+\nu^{2}\right) y(x)=0,   
\end{equation}
with singularities at $0$ and $\infty$. The solutions to this equation define the Bessel functions $I_{\nu}(x)$ \cite{abramowitz1968handbook}. If $\nu$ is not an integer, then $I_{\nu}(x)$ and $I_{-\nu}(x)$ are two linearly independent solutions. If $\nu=n$, an integer, it holds that $I_{-n}(x)=I_{n}(x)$. Another solution to Eq.~\eqref{eq:edo} is the modified Bessel function of the second kind, $K_{\nu}(x)$, which can be expressed as:
\begin{equation} \label{eq:Ibessel}
 K_{\nu}(x)=\frac{\pi}{2} \frac{I_{-\nu}(x)-I_{\nu}(x)}{\sin (\nu \pi)} .
\end{equation}

\begin{Definition}\label{def:hipergometric} The generalized hypergeometric function, ${}_pF_q$, is given by
\begin{equation}
f(x)  ={ }_{p} F_{q}\left(\left.\begin{array}{c}
a_{1}, \ldots, a_{p} \\
b_{1}, \ldots, b_{q}
\end{array} \right\rvert\, x\right)  =\sum_{n=0}^{\infty} \frac{\left(a_{1}\right)_{n} \ldots\left(a_{p}\right)_{n}}{\left(b_{1}\right)_{n} \ldots\left(b_{q}\right)_{n}} \frac{x^{n}}{n !},
\end{equation}
where the indexes $\{p, q\} \in(\mathbb{N}+\{0\})$, the parameters $a_{j} \in \mathbb{R}(j=1, \ldots, p)$ y $b_{k} \in \mathbb{R}(k=1, \ldots, q)$, and 
\begin{equation}
(\beta)_{n}=\frac{\Gamma(\beta+n)}{\Gamma(\beta)},
\end{equation}
is known as the symbol of Pochhammer and the convergence conditions for the hypergeometric function are $p>q+1$, $p=q+1$, and $p<q+1$, which correspond to a radius of convergence $R$ equal to $0$, $1$, and $\infty$, respectively.
\end{Definition}

\begin{theorem}\label{teo:seriebracketsBessel}
The Brackets series representations of $K_{\nu}(x)$ Bessel function is given by
\begin{equation}\label{eq:besselseriebracket}
K_{\nu}(x)=2^{\nu} \sum_{n_{1} \geq 0} \sum_{n_{2} \geq 0} \sum_{n_{3} \geq 0} \phi_{n_{1}, \ldots, n_{3}} \frac{x^{\nu+2 n_{3}}}{\Gamma\left(\frac{1}{2}+n_{1}\right) 4^{n_{1}}}\left\langle\nu+\frac{1}{2}+n_{2}+n_{3}\right\rangle\left\langle 2 n_{1}+2 n_{2}+1\right\rangle.
\end{equation}
\end{theorem}
This representation will be used later to evaluate the synchrotron function.
\begin{proof}
The function $K_{\nu}(x)$ is a singular function at $x=0$, which has no series representation around $x=0$, but two integral representations of the form:
\begin{equation}\label{eq:besselintegral2}
 K_{\nu}(x)=\frac{\sqrt{\pi}}{\Gamma\left(\nu+\frac{1}{2}\right)}\left(\frac{x}{2}\right)^{\nu} \int_{1}^{\infty} \exp (-x t)\left(t^{2}-1\right)^{\nu-\frac{1}{2}} d t , 
\end{equation}
and
\begin{equation} \label{eq:besselintegral}
 K_{\nu}(x)=\frac{2^{\nu} x^{\nu} \Gamma\left(\nu+\frac{1}{2}\right)}{\sqrt{\pi}} \int_{0}^{\infty} \frac{\cos (t)}{\left(t^{2}+x^{2}\right)^{\nu+\frac{1}{2}}} d t.
\end{equation}
In the integrand of Eq.~\eqref{eq:besselintegral}, the function $\cos(t)$ has a hypergeometric representation given in \cite{abramowitz1968handbook} as:
\begin{equation}
\cos (t)  ={ }_{0} F_{1}\left(\left.\begin{array}{c}
- \\
\frac{1}{2}
\end{array} \right\rvert\,-\frac{1}{4} t^{2}\right)  =\sqrt{\pi} \sum_{n_{1} \geq 0} \phi_{n_{1}} \frac{1}{\Gamma\left(\frac{1}{2}+n_{1}\right) 4^{n_{1}}} t^{2 n_{1}},
\end{equation}
For the denominator of the binomial, we apply 
rule \textbf{R3},
obtaining
\begin{equation}
    \frac{1}{\left(t^{2}+x^{2}\right)^{\nu+\frac{1}{2}}}=\sum_{n_{2} \geq 0} \sum_{n_{3} \geq 0} \phi_{n_{2}} \phi_{n_{3}} t^{2 n_{2}} x^{2 n_{3}} \frac{\left\langle\nu+\frac{1}{2}+n_{2}+n_{3}\right\rangle}{\Gamma\left(\frac{1}{2}+\nu\right)}.
\end{equation}
Replacing the last two expansions in Eq.~\eqref{eq:besselintegral}, we evaluate the integral as indicated in Section~\ref{sec:formalism}, demonstrating that the Bessel function \(K_\nu\) corresponds to the bracket series representation in Eq.~\eqref{eq:besselseriebracket}.
\end{proof}

\subsection{Asymptotic Expansion of $K_{\nu}(x)$ for large $x$}

In this subsection, the MoB will be utilized to obtain the behavior as $x \rightarrow \infty$ of the Bessel function. By making the following change $u = t - 1$, the integral in Eq.~\eqref{eq:besselintegral2} transforms into
 \begin{equation}\label{eq:besselintegral3}
 K_{\nu}(x)=\frac{\sqrt{\pi}}{\Gamma\left(\nu+\frac{1}{2}\right)}\left(\frac{x}{2}\right)^{\nu} \exp (-x) \int_{0}^{\infty} \exp (-x u) u^{\nu-\frac{1}{2}}(2+u)^{\nu-\frac{1}{2}} d u.
 \end{equation}

\begin{note} \label{nota:exp}
The exponential function can be written as
\begin{equation}
\exp (-x u)=\sum_{n_{1} \geq 0} \phi_{n_{1}} x^{n_{1}} u^{n_{1}}.
\end{equation}
\end{note}
\begin{note}\label{nota:polinomio}
The expression $(2+u)^{\nu-\frac{1}{2}}$, using rule \textbf{R3},
can be expressed as:
\begin{equation}
    \sum_{n_{2} \geq 0} \sum_{n_{3} \geq 0} \phi_{n_{2}} \phi_{n_{3}} 2^{n_{2}} u^{n_{3}} \frac{\left\langle\frac{1}{2}-\nu+n_{2}+n_{3}\right\rangle}{\Gamma\left(\frac{1}{2}-\nu\right)}.
\end{equation}
\end{note}
Expressing the Bessel function from Eq.~\eqref{eq:besselintegral3} in the expansions from notes \ref{nota:exp} and \ref{nota:polinomio}, we obtain the following bracket series:

\begin{equation}
 \begin{aligned}
K_{\nu}(x) = & \ \frac{\sqrt{\pi}}{\Gamma\left(\nu+\frac{1}{2}\right) \Gamma\left(\frac{1}{2}-\nu\right)}\left(\frac{x}{2}\right)^{\nu} \exp (-x) \\
& \times \sum_{n_{1} \geq 0} \sum_{n_{2} \geq 0} \sum_{n_{3} \geq 0} \phi_{n_{1}, \ldots, n_{3}} 2^{n_{2}} x^{n_{1}}\left\langle\frac{1}{2}-\nu+n_{2}+n_{3}\right\rangle\left\langle\nu+\frac{1}{2}+n_{1}+n_{3}\right\rangle.
\end{aligned}   
\end{equation}

Note that the number of sums exceeds the number of brackets; therefore, the summation rule will provide us with three expressions, which can be written as
\begin{equation}
T_{1} (x)=\frac{2^{\nu} x^{\nu} \cos (\nu \pi) \exp (-x)}{\sqrt{\pi}} \sum_{n \geq 0} \frac{\Gamma\left(\nu+\frac{1}{2}+n\right) \Gamma(-2 \nu-n)}{n !}(-2 x)^{n} = -\frac{\pi}{2} \frac{I_{\nu}(x)}{\sin (\pi \nu)},
\end{equation}
\begin{equation}
T_{2}  (x)=\frac{2^{-\nu} x^{-\nu} \cos (\nu \pi) \exp (-x)}{\sqrt{\pi}} \sum_{n \geq 0} \frac{\Gamma\left(-\nu+\frac{1}{2}+n\right) \Gamma(2 \nu-n)}{n !}(-2 x)^{n} =\frac{\pi}{2} \frac{I_{-\nu}(x)}{\sin (\pi \nu)},
\end{equation}
where $T_{1}(x)$ and $T_{2}(x)$ are described as confluent hypergeometric functions, and 
\begin{equation}\label{eq:2F0}
\begin{aligned}
T_3 (x) = & \ \frac{\cos (\nu \pi) \exp (-x)}{\sqrt{2 \pi x}} \sum_{n \geq 0} \frac{\Gamma\left(-\nu+\frac{1}{2}+n\right) \Gamma\left(\nu+\frac{1}{2}+n\right)}{n!}\left(-\frac{1}{2 x}\right)^n \\
 = &\  \sqrt{\frac{\pi}{2 x}} \exp (-x)\pFq20{\tfrac{1}{2} -\nu  \quad   \tfrac{1}{2} +\nu } {-}{ -\frac{1}{2x}} .
\end{aligned}
\end{equation}

The solution of the integral in Eq.~\eqref{eq:besselintegral3} is obtained by summing all the terms with the same region of convergence. In this case, it can be observed that:
\begin{equation}
K_{\nu}(x)=\left\{\begin{array}{l}
T_{1}(x)+T_{2}(x) \\
T_{3}(x)
\end{array}\right.    
\end{equation}
The first solution corresponding to arguments \(x^n\) corresponds to the expression in Eq.~\eqref{eq:Ibessel}. The MoB also finds a second solution for argument \(x^{-1}\) through a divergent hypergeometric representation, \textit{i.e.}, a convenient representation for the asymptotic behavior of $x$.
This final representation will be useful later to obtain the approximation of the function $F(x)$ for very large values of $x$.
\begin{theorem} \label{theo:Kbessell}
The divergent representation of $K_{\nu}(x)$ is given by
\begin{equation}
K_{\nu}(x)=\sqrt{\frac{\pi}{2 x}} \exp (-x)\ \pFq20{\frac{1}{2} -\nu \quad \frac{1}{2} +\nu }{-}{-\frac{1}{2x}}.
\end{equation}
\end{theorem}
The behavior for large $x$, utilizing the divergent representation, yields a perturbative analytical result for the Bessel function, expressed as follows:
\begin{equation}
\begin{aligned}
K_{\nu}(x) =  \sqrt{\frac{\pi}{2 x}} \exp (-x)\left[1+\frac{4 \nu^{2}-1}{8 x}+\frac{16 \nu^{4}-40 \nu^{2}+9}{128 x^{2}}+ \mathcal{O} \left(\frac{1}{x^3}\right)\right] .
\end{aligned}
\end{equation}
Numerically, in Figure~\ref{fig:seriebessel}, the solid blue line represents the graph of the modified Bessel function \(K_\nu(x)\), while the gray points illustrate the behavior of the asymptotic series derived from the divergent representation obtained via 
%
MoB,
both for \(\nu=5/3\), which is of particular interest in this work. Note that our fit coincides with the modified Bessel function within the presented region.

\begin{figure}
    \centering
    \includegraphics[scale=0.65]{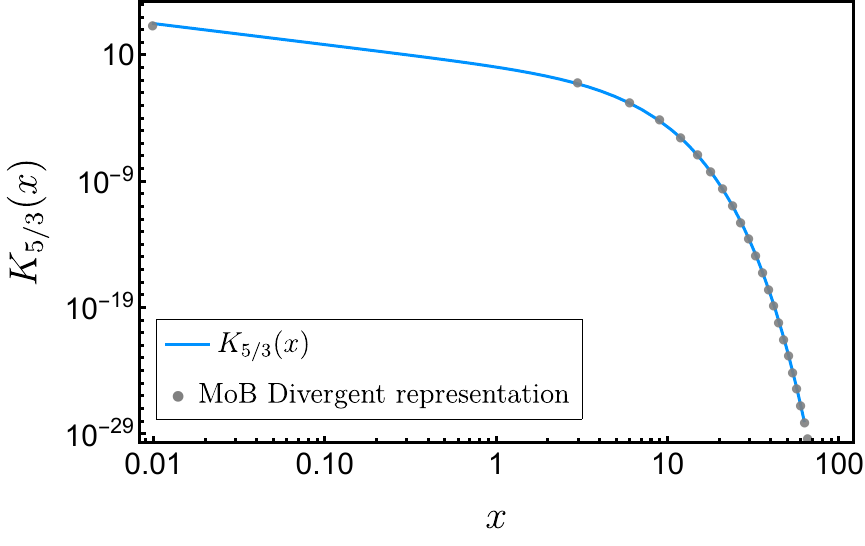}
    \caption{Form of the divergent representation of the modified Bessel function \(K_\nu(x)\) for \(\nu=5/3\), derived from the bracket series. This asymptotic representation demonstrates convergence with numerically calculated values.}
    \label{fig:seriebessel}
\end{figure}

\section{Analytical Solution of the Synchrotron Function}
\label{sec:solution}

In this section we obtain an \emph{analytical} expression for the synchrotron function and its asymptotic limits using MoB.

The synchrotron function, which describes the radiation emitted by charged particles accelerated in magnetic fields, is defined through an integral involving the modified Bessel function. For a particle with charge $e$, energy $E=\gamma mc^{2}$, and velocity $v=\beta c$, traveling in a circular orbit of radius $R$, the energy loss per revolution is particularly significant for high‑energy electrons or positrons, where $\beta\approx1$. In this case, the energy loss per revolution is approximately $\delta E\simeq0.0885\,\text{MeV}/R$. This rapid increase in energy loss strongly depends on the particle’s energy. The energy radiated per revolution into the photon‑energy interval is
\begin{equation}
dI = \frac{4\pi}{3}\,\alpha\,\gamma\,F\!\left(\frac{\omega}{\omega_c}\right)\,d\!\left(\hbar\omega\right),
\end{equation}
where $\alpha$ is the fine‑structure constant, $F(\omega/\omega_{c})$ is the dimensionless \emph{synchrotron function}, and $\omega_{c}$ is the critical frequency, defined as
\begin{equation}
\omega_c = \frac{3\gamma^{3}c}{2R}.
\end{equation}

Next, the normalized synchrotron function $F(y)$ from Eq.~\eqref{eq:integral} has been calculated for an arbitrary order $\nu$ of the modified Bessel function. With the substitution $u = t - x$, the integral can be written as
\begin{equation}
S(x)=\int_x^{\infty} K_\nu(t) d t=\int_0^{\infty} K_\nu(u+x) d u.
\end{equation}

Since theorem
\ref{teo:seriebracketsBessel} and using rule \textbf{R3} for residual binomial $(u+x)$,  the integral $S(x)$ is described by the following brackets series:
\begin{equation}
S(x)= \sum_{n_1} \ldots \sum_{n_5} \phi_{{n_1} \ldots \phi_{n_5}} \frac{2^\nu x^{n_5}}{4^{n_1}} \frac{\left\langle\nu+\frac{1}{2}+n_2+n_3\right\rangle 
\left\langle 2 n_1+2 n_2+1\right\rangle
\left\langle-\nu-2 n_3+n_4+n_5\right\rangle\left\langle n_4+1\right\rangle}
{ \Gamma\left(\frac{1}{2}+n_1\right) \Gamma\left(-\nu-2 n_3\right)}.
\end{equation}

From the previous series, we can construct the following analytical terms for the same region of convergence:
\begin{equation}
S_1(x) = \left(\frac{x}{2} \right)^{1-\nu} \Gamma(\nu-1)  \ \pFq12{\frac{1}{2} -\frac{\nu}{2}}{1 -\nu \quad \frac{3}{2}-\tfrac{\nu}{2}}{\frac{x^2}{4}},
\end{equation}
\begin{equation}
S_2(x) = \left(\frac{x}{2} \right)^{1+\nu} \Gamma(-\nu-1)  \ \pFq12{\frac{1}{2} +\frac{\nu}{2}}{1 +\nu \quad \tfrac{3}{2}+\tfrac{\nu}{2}}{\frac{x^2}{4}},
\end{equation}
and 
\begin{equation}
S_3(x) = \frac{1}{2}\sum_{k =0}^{\infty} \frac{(-1)^k}{k!} \frac{\Gamma\left(\frac{1}{2}+ \frac{\nu}{2}-\frac{k}{2}\right) \Gamma\left(\frac{1}{2}- \frac{\nu}{2}-\frac{k}{2}\right)}{\Gamma\left(1-k\right) } \left( \frac{x}{2}\right)^k,
\end{equation}
where the solution of the integral $S(x)$ corresponds to $S(x)= S_1(x) + S_2(x) + S_3(x)$. It is observed that, of the series $S_3(x)$ only the first term is non-zero, that is
\begin{equation}
S_3(x)=\frac{\Gamma\left(\frac{1}{2} + \frac{\nu}{2}\right)\Gamma\left(\frac{1}{2} -\frac{\nu}{2}\right)}{2}
\end{equation}

At the same time, the solution for the reciprocal argument is given by the following series.
\begin{equation}
S_4(x)=2^\nu x^{-\nu} \sum_{k=0}^{\infty} \frac{\Gamma(\nu+2 k) \Gamma\left(\nu+\frac{1}{2}+k\right) \Gamma\left(\frac{1}{2}+k\right)}{\Gamma(-k) \Gamma(\nu+1+2 k) k!}\left(-\frac{4}{x^2}\right)^k,
\end{equation}
where the series vanishes for all values of the index $k$. 
The contribution of this series is not of interest for the present study, since it belongs to a different region of convergence (argument $\sim x^{-1}$).

\bigskip 

By evaluating the sum $S$ with $\nu=5/3$ and keeping only the physically relevant (nonzero) terms, one obtains the analytic form of the synchrotron function, which can be expressed as
\begin{equation}\label{eq:Hsmall}
F(x) = -\frac{ \pi }{\sqrt{3}} x +  
(4x)^{1/3} \Gamma\left( \frac{2}{3}\right) \ \pFq12{-\frac{1}{3} }{-\frac{2}{3} \quad \frac{2}{3}}{\frac{x^2}{4}}  +\left( \frac{x^{11}}{256}\right)^{1/3}\Gamma \left(-\frac{8}{3}\right)  \ \pFq12{\frac{4}{3}}{\frac{8}{3} \quad \frac{7}{3}}{\frac{x^2}{4}},
\end{equation}
this result is valid for $0 \le x < \infty$. The analytical form obtained through the Method of Brackets shows excellent agreement with the numerical integration of the synchrotron function. As illustrated in Fig.~\ref{fig:Fx}, the MoB solution (gray dots) reproduces the numerical behavior (green line) across the entire range of $x$, confirming the accuracy and robustness of the proposed analytical approach.

\begin{figure}
\centering
\includegraphics[width=0.65\linewidth]{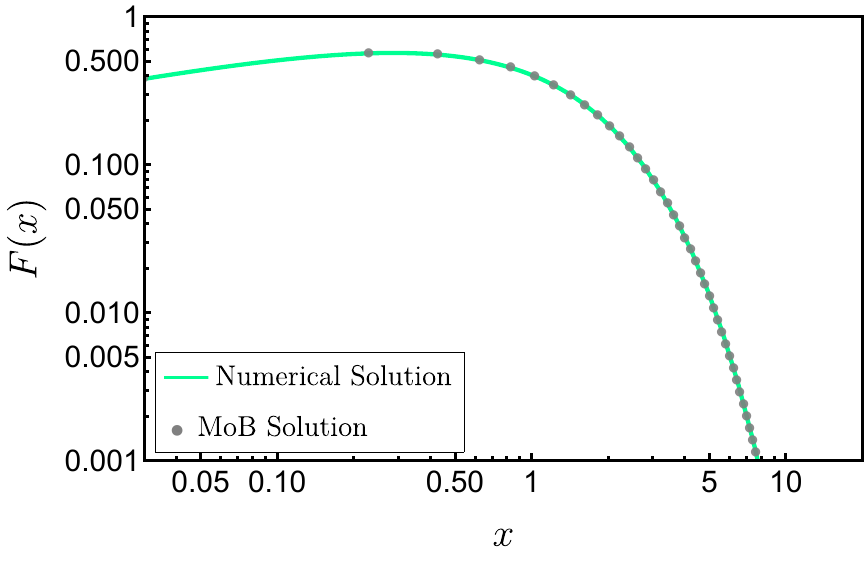}
\caption{Analytical (gray dots) and numerical (green line) results for the synchrotron function $F(x)$ over the entire range of $x$.}
\label{fig:Fx}
\end{figure}

\section{Small-$x$ asymptotics of $F(x)$}
\label{sec:small}

In this section we describe the regime $x \ll 1$ of the synchrotron function $F(x)$, 
establishing its expansion in fractional powers and the hierarchical order of corrections 
that emerge from the local behavior of $K_{5/3}$ near the origin.

The resulting asymptotic expansion of $F(x)$ around $x=0$ can be written explicitly as follows:
\begin{equation}\label{eq:Fsmall}
\begin{aligned}
F(x) 
= & \ 
2^{2/3} \Gamma\left(\frac{2}{3}\right) x^{1/3}
- \frac{\pi}{\sqrt{3}} x 
+ \frac{3 \Gamma\left(\frac{2}{3}\right)}{4^{5/3}} x^{7/3}
+ \mathcal{O}\left( x^{11/3} \right),
\end{aligned}
\end{equation}
this result is obtained directly from Eq.~\eqref{eq:Hsmall}.
This expansion clearly exhibits the non-integer power structure that characterizes the small-$x$ regime of $F(x)$. 

\medskip

To quantify the accuracy of the small-$x$ expansion, 
Table~\ref{tab:smallx} presents the relative error 
$\varepsilon_{\text{rel}}(x,M)$ for representative values of $x$ 
and different truncation orders $M$, where $S_M(x)$ denotes the truncated sum up to a given term. 
As expected for an asymptotic series, the error decreases systematically as more terms are included, 
provided that $x$ remains sufficiently small. 
This analysis confirms that truncation at the $x^{11/3}$ or $x^{13/3}$ term 
already yields sub-percent accuracy for $x \lesssim 10^{-2}$.

\begin{table}[t]
\centering
\begin{tabular}{ccccc}
\toprule
$x$ & Terms $M$ & $S_M(x)$ & $F_{\text{exact}}(x)$ & $\varepsilon_{\text{rel}}(x,M)$ \vspace{0.1cm} \\ 
\midrule
0.001 & 1 & 0.214953 & 0.213139 & $8.51\times 10^{-3}$ \\
0.001 & 2 & 0.213139 & 0.213139 & $1.89\times 10^{-7}$ \\
0.001 & 3 & 0.213139 & 0.213139 & $6.65\times 10^{-12}$ \\
0.001 & 4 & 0.213139 & 0.213139 & $2.84\times 10^{-14}$ \\
0.001 & 5 & 0.213139 & 0.213139 & $3.57\times 10^{-19}$ \\
0.01  & 1 & 0.463102 & 0.444973 & $4.07\times 10^{-2}$ \\
0.01  & 2 & 0.444964 & 0.444973 & $1.95\times 10^{-5}$ \\
0.01  & 3 & 0.444973 & 0.444973 & $1.46\times 10^{-8}$ \\
0.01  & 4 & 0.444973 & 0.444973 & $2.93\times 10^{-10}$ \\
0.01  & 5 & 0.444973 & 0.444973 & $7.84\times 10^{-14}$ \\
0.05  & 1 & 0.791893 & 0.701572 & $1.29\times 10^{-1}$ \\
0.05  & 2 & 0.701203 & 0.701572 & $5.26\times 10^{-4}$ \\
0.05  & 3 & 0.701574 & 0.701572 & $3.25\times 10^{-6}$ \\
0.05  & 4 & 0.701572 & 0.701572 & $1.98\times 10^{-7}$ \\
0.05  & 5 & 0.701572 & 0.701572 & $4.42\times 10^{-10}$ \\
\bottomrule
\end{tabular}
\caption{Relative error of the small-$x$ asymptotic expansion of $F(x)$ 
for several truncations $S_M(x)$.}
\label{tab:smallx}
\end{table}

\medskip

Equation~\eqref{eq:Fsmall} provides the asymptotic expansion of $F(x)$ for $x \ll 1$, 
yielding a compact analytic description of the structure inherited from $K_{5/3}$, 
which does not admit a regular expansion around the origin. 
For practical computations in the small-$x$ regime, it is sufficient to truncate the series at the required order.

\section{Large-$x$ asymptotics of $F(x)$}
\label{sec:large}

An analytical expression for the asymptotic behavior of the synchrotron function $F(x)$ in the large-argument regime, $x \gg 1$, is derived below. The analysis developed in this work provides a systematic procedure that allows obtaining the expansion to any desired order.

To derive the solution in the large-$x$ regime, we employ the divergent representation of the modified Bessel function $K_{\nu}(x)$ (see Theorem~\ref{theo:Kbessell}\footnote{The integration of divergent series representations using MoB can be reviewed in Ref. \cite{gonzalez2017extension}.}) within the synchrotron function $F(x)$. 
This substitution reduces the problem to evaluating the following integral:
\begin{equation}
S(x)=\int_x^{\infty} K_\nu(t) d t=\sqrt{\frac{\pi}{2}} \int_x^{\infty} \frac{e^{-t}}{\sqrt{t}}\ \pFq20{\tfrac{1}{2} -\nu \quad \tfrac{1}{2} +\nu }{-}{-\frac{1}{2t}} \ dt.
\end{equation}

According to Definition~\ref{def:hipergometric}, the previous expression can be rewritten as
\begin{equation}\label{eq:Jlarge}
S(x)=\sqrt{\frac{\pi}{2}}\sum_{n \geq 0}  \phi_n 
\frac{\Gamma\left(\frac{1}{2} + \nu + n\right)\Gamma\left(\frac{1}{2} - \nu + n\right)}{2^n\Gamma\left(\nu + \frac{1}{2}\right)\Gamma\left(\frac{1}{2} - \nu\right)}
\int_x^{\infty} \frac{e^{-t}}{t^{n + \frac{1}{2}}} dt.
\end{equation}

To apply MoB, we first transform the integral in Eq.~\eqref{eq:Jlarge} by introducing the change of variable $u = t - x$. 
This substitution shifts the integration limits to the interval $[0, \infty)$, allowing the use of the Integration Rule~(\textbf{R2}), \textit{i.e.}, 
\begin{equation}
\int_x^{\infty} \frac{e^{-t}}{t^{n + \frac{1}{2}}} dt = \ift \frac{e^{-x-u}}{(x+u)^{n+\frac{1}{2}}} du.
\label{eq:int4}
\end{equation}

Taking into account Note~\ref{nota:exp} and the multinomial expansion~(\textbf{R3}) in the LHS integral of Eq.~\eqref{eq:int4}, the bracket series representation of Eq.~\eqref{eq:Jlarge} takes the form
as follows:
\begin{equation}
S(x) = \sqrt{\frac{\pi}{2}} \sum_{n\geq 0} \sum_{m_1\geq 0}  \sum_{m_2\geq 0}  \sum_{m_3\geq 0} \phi_{n,m_1,m_2,m_3} 
c_{\nu}(n)\ x^{m_3} 
\left\langle m_1+m_2 +1\right\rangle
\left\langle n+m_2+m_3 + \frac{1}{2}\right\rangle,
\label{eq:jLarge1}
\end{equation}
where
\begin{equation}
c_{\nu}(n) =
\frac{\Gamma \left( \frac{1}{2}-\nu+n\right) \Gamma \left( \frac{1}{2}+\nu+n \right)}
{2^{n}\Gamma \left( \frac{1}{2}-\nu
\right) \Gamma \left( \frac{1}{2}+\nu \right) \Gamma\left( \frac{1}{2} + n\right) }.
\end{equation}

The bracket series in Eq.~\eqref{eq:jLarge1} yields five formal independent terms, whose explicit expressions are collected in Appendix~\ref{sec:app1}. 
A convergence analysis (summarized in Table~\ref{tab:brackets-large-x}) shows that, in the large-$x$ regime, only $S_2(x)$ provides the relevant asymptotic representation as $x\to\infty$ (in the sense of optimal truncation), whereas the remaining solutions are either restricted to a different domain of convergence.

In the large-$x$ regime, we set $F(x)=x\ S_2(x)$, this is
\begin{equation}
    F(x) = \  \sqrt{\frac{\pi x}{2}} e^{-x} \left( \sum_{n=0}^{\infty}\sum_{m=0}^{\infty}\frac{\Gamma \left( \tfrac{1}{2}%
+m+n\right) \Gamma \left( \tfrac{1}{2}+\nu +m\right) \Gamma \left( \tfrac{1}{%
2}-\nu +m\right) }{2^{m}  m! \Gamma \left( \frac{1}{2}%
+\nu \right) \Gamma \left( \frac{1}{2}-\nu \right)\Gamma \left( \tfrac{1}{2}+m\right)}\left(- \frac{1}{x}\right) ^{m} \left( -\frac{%
1}{x}\right) ^{n}\right)  .
\end{equation}

For the particular case with $\nu=5/3$ and truncate the resulting double series at order $N$ in each summation index, thereby obtaining the asymptotic approximation.
\begin{equation}
F\left( x\right) =  \sqrt{\frac{x}{\pi}}
e^{-x}
 \left( 
 \sum_{n=0}^{N} \sum_{m=0}^{N} \frac{
 \Gamma\left(\frac{1}{2} +n+m \right) \Gamma\left(-\frac{7}{6}+m \right) \Gamma\left(\frac{13}{6} +m \right)
 }{2^{3/2+m} \Gamma\left(\frac{1}{2}+m \right)}
 \left(- \frac{1}{x} \right)^m \left(- \frac{1}{x} \right)^n 
 \right).
\end{equation}

Equivalently, by expanding the truncated double sum in inverse powers of $x$ and factoring out the dominant exponential prefactor, one obtains the standard large-$x$ asymptotic series, in particular for $N=2$, we obtain 
\begin{equation}
F(x) = \sqrt{\frac{\pi x}{2}} e^{-x}\Bigg[
1+\frac{55}{72} \frac{1}{x}
-\frac{10151}{10368} \frac{1}{x^2}+ \mathcal{O}\left( \frac{1}{x^3}\right)
\Bigg],
\end{equation}
where the truncation order is taken to be of order $N$.

\bigskip

\begin{table}[t]
\centering
\label{tab:brackets-large-x}
\begin{tabular}{c c l}
\toprule
Solution & Argument & Validity region  \\
\midrule
$S_1(x)$  & $x^{n-m}$                    & Valid for $x \to \infty$. \\
$S_2(x)$   & $x^{-n-m}$                   & Valid for $x \to \infty$, optimal truncation. \\
$S_3(x)$  & $x^{n}$                                & Valid for $0 < x < \infty$.\\
$S_4(x)$ & $x^{n}$                      & Valid for $0 < x < \infty$. \\
$S_5(x)$ & $x^{n}$         & Null. \\
\bottomrule
\end{tabular}
\caption{Summary of the independent solutions derived from the bracket series and and their corresponding domains of validity.}
\end{table}

Equations above provide a systematically improvable large-$x$ approximation of $F(x)$ with a controllable remainder, making the evaluation of the synchrotron function efficient in the high-energy tail. 
Together with the small-$x$ expansion derived previously, this completes an analytic description of $F(x)$ in both asymptotic regimes, suitable for practical implementations where repeated numerical quadratures are undesirable.

\section{\label{sec:conclusions}Conclusions}

In this work we derived an analytical representation of the synchrotron function by applying the MoB to its defining integral. Starting from a bracket-series representation of the modified Bessel function, the method reduces the problem to a finite set of formal series solutions whose relevance is determined by their convergence properties. This yields a compact closed-form expression for $F(x)$ and, at the same time, provides a systematic route to asymptotic approximations.

An important outcome is that the MoB approach delivers controlled and systematically improvable expansions in the two regimes of practical interest. In the small-$x$ limit, we obtain the expected fractional-power structure and quantify the accuracy of truncated expansions through explicit error estimates. In the large-$x$ limit, we construct an exponentially suppressed asymptotic series and show that optimal truncation leads to an efficient approximation in the high-energy tail.

Our analytical results were validated against direct numerical integration, showing very good agreement in the representative ranges considered. Overall, the formulas presented here offer a practical alternative to repeated numerical quadratures and help make the analytic structure of the synchrotron kernel explicit. The same strategy can be extended to related synchrotron kernels and to integrals involving $K_\nu$ that appear in radiation and transport calculations.

\section*{Acknowledgments}
\noindent Ivan Gonzalez would like to thank CEFITEV-UV for partial support. 

\appendix

\section{Complete solution for large-$x$ values}
\label{sec:app1}

From the bracket series of Eq.~\eqref{eq:jLarge1}, one finds the following solutions:
\begin{equation}
\begin{aligned}
S_1(x) = & \  \sqrt{\frac{\pi x}{2}} e^{-x} \sum_{n=0}^{\infty} \sum_{m=0}^{\infty} 
\frac{c_{\nu}(m)}{m!} \Gamma\left(-\frac{1}{2} -n+m\right)  \left(- \frac{1}{x} \right)^m \left(- x \right)^n ,\\
S_2(x) = &  \  \sqrt{\frac{\pi }{2x}} e^{-x} \sum_{n=0}^{\infty} \sum_{m=0}^{\infty} \frac{ c_{\nu}(m)}{m!}
\Gamma\left(\frac{1}{2} + n+m\right)  \left(- \frac{1}{x} \right)^m \left(- \frac{1}{x} \right)^n, \\
S_3(x) = &  \  \sqrt{\frac{\pi }{2}} e^{-x} \sum_{n=0}^{\infty} \sum_{m=0}^{\infty} \frac{(-1)^{m} c_{\nu}(m)}{n! m!}
\Gamma\left(\frac{1}{2} + n+m\right) \Gamma\left(\frac{1}{2} - n-m\right) \left(-x \right)^n, \\
S_4(x) = &  \  \frac{\sqrt{\pi}}{2} e^{-x} \sum_{n=0}^{\infty} \sum_{m=0}^{\infty}
\frac{(-1)^{m} }{2^{m}n! } 
\frac{\Gamma\left(-\frac{1}{2} +n-m\right) \Gamma\left(1+\nu-n+m\right)\Gamma\left(1-\nu-n+m\right)}
{ \Gamma\left(1-n+m\right) \Gamma\left(\frac{1}{2} -\nu\right) \Gamma\left(\frac{1}{2}+\nu\right)  } \left(-2x\right)^{n},\\
S_5(x) = &  \  \sqrt{\pi} e^{-x} \sum_{n=0}^{\infty} \sum_{m=0}^{\infty}
\frac{(-2)^{m} }{n! } 
\frac{\Gamma \left( \tfrac{1}{2}%
+m+n\right) \Gamma (\nu -m-n)\Gamma (-\nu -m-n)}{\Gamma
(-m-n) \Gamma \left( \frac{1}{2}-\nu
\right) \Gamma \left( \frac{1}{2}+\nu \right) } \left(-2x \right)^n.
\end{aligned}
\end{equation}

The convergence analysis shows that $S_1$, $S_3$ and $S_4$ correspond to the same convergence region. While the term $S_2$ involves two divergences summations that both define asymptotic expansions around $x \to \infty$, it therefore provides the optimal approximation in the large-$x$ regime.

\bibliographystyle{utphys} 
\bibliography{references.bib} 

\end{document}